\documentclass[10pt]{article}

\usepackage[a4paper, total={6in, 9in}]{geometry}
\usepackage{amsmath, graphicx,epsfig ,enumerate,latexsym,amssymb, amsthm}
\usepackage{bm}
\usepackage{color}
\usepackage[ruled]{algorithm2e}
\usepackage{sansmath}
\usepackage{comment}
\usepackage{authblk}
\usepackage{subcaption}
\usepackage{tcolorbox}
\usepackage{tikz}
\usetikzlibrary{arrows,positioning,automata}
\usepackage{hyperref}

\SetKwRepeat{Do}{do}{while}

\newtheorem{theorem}{Theorem}

\newcommand{\ignore}[1]{}

\DeclareMathOperator*{\argmax}{arg\,max}
\newcommand{\algo}[1]{\ensuremath{\mathsf{#1}}}

\newcommand{\defn}[1]{{\textbf{\emph{#1}}}}

 \newcommand{\ra}[1]{{}}
 \newcommand{\gs}[1]{{}}
 \newcommand{\bs}[1]{{}}
 \newcommand{\bsv}[1]{{}}

\begin{document}


\title{Don't Be Greedy: Leveraging Community Structure to Find High Quality Seed Sets for Influence Maximization\thanks{This research was supported by the National Science Foundation CAREER Award \#1452915 and National Science Foundation AitF Award \#1535912}}
\author[1]{Rico Angell}
\author[1]{Grant Schoenebeck}
\affil[1]{University of Michigan}

\date{\today}

\maketitle

\begin{abstract}
We consider the problem of maximizing the spread of influence in a social
  network by choosing a fixed number of initial seeds --- a central problem in the study of network cascades.
  The majority of existing work on this problem, formally referred to as the
  \emph{influence maximization problem}, is designed for  submodular cascades.
  Despite the empirical evidence that many cascades are non-submodular, little work has been done focusing on non-submodular influence maximization.


We propose a new heuristic for solving the influence maximization problem and show via simulations on real-world and synthetic networks that our algorithm outputs more influential seed sets than the state-of-the-art greedy algorithm in many natural cases, with average improvements of 7\% for submodular cascades, and 55\% for non-submodular cascades.  
Our heuristic uses a dynamic programming approach on a hierarchical decomposition of the social network to leverage the relation between the spread of cascades and the
community structure of social networks.   We verify the importance of network
structure by showing  the quality of the hierarchical decomposition impacts the
quality of seed set output by our algorithm.
We also present ``worst-case" theoretical results proving that in  certain
settings our algorithm outputs seed sets that are a factor of $\Theta(\sqrt{n})$
more influential than those of the greedy algorithm, where $n$ is the number of
nodes in the network.   Finally, we generalize our algorithm to a message passing
version that can be used to find seed sets that have at least as much influence as the dynamic programming algorithms.  



\end{abstract}

%
%
%
\section{Introduction}
\label{sec:intro}

A \emph{cascade} is a fundamental social network process in which a number of nodes, or agents, start with some property that they then may spread to neighbors.  Network structure has been shown  relevant for a wide array of real world cascade processes including the adoption of products~\cite{Brown87}, 
farming technology~\cite{ConleyU10}, medical practices~\cite{ColemanKM57}, participation in
microfinancing~\cite{BanerjeeCDJ13}, and the spread of information over social networks~\cite{LermanG10}.

How to place a limited number of initial seeds, in order to maximize the spread
of the resulting cascade, is a natural question known as   \textsc{InfluenceMaximization}~\cite{Domingos01,RichardsonD02,KempeKT03,KempeKT05,MosselR10}.   This problem requires as input a network, a cascade process, and the number of initial seeds.  For example, which students can most effectively be enrolled in an intervention to decrease student conflict at a school~\cite{PaluckSA2016}?

To study \textsc{InfluenceMaximization}, we first need to understand how cascades spread.  While many cascade models have been proposed~\cite{Arthur89,Morris00,Watts02}, they can be roughly divided into two categories: \emph{submodular} and \emph{non-submodular}.

In submodular cascade models, such as the Independent Cascade model defined in Section~\ref{sec:prelim}~\cite{KempeKT03,KempeKT05,MosselR10}, a node's marginal probability of becoming infected after a new neighbor is infected decreases when the number of previously infected neighbors increases~\cite{KempeKT03}.  
In non-submodular cascade models the marginal probability of being infected may increase as more neighbors are infected.  For example, in the Threshold model~\cite{Granovetter78}, each node has a threshold for the number of infected neighbors after which it too will become infected.  If a node has a threshold of 2, then the first infected neighbor has zero marginal impact, but the second infected neighbor causes this node to become infected with probability 1.  Unlike submodular cascades, non-submodular cascades require well-connected regions to spread~\cite{Centola10}.

For \textsc{InfluenceMaximization}  in submodular cascades,  a straightforward greedy algorithm  efficiently finds a seed set with influence at least a $(1- 1/e)$ fraction of the optimal; but for general non-submodular casacdes, it is NP-hard even to approximate \textsc{InfluenceMaximization} to within a $n^{1-\epsilon}$  factor of optimal~\cite{KempeKT03}.   

Unfortunately, empirical research shows that most cascades are non-submodular~\cite{RomeroMK11,BackstromHKL06,LeskovecAH06}, and in this case little is known about \textsc{InfluenceMaximization} other than worst-case hardness.  \textsc{InfluenceMaximization} becomes qualitatively different in the non-submodular setting.
In the submodular case, one should put as much distance between the $k$ initial
adopters as possible, lest they erode each other's effectiveness.  However, in
the non-submodular case, it may be advantageous  to place the initial adopters
close together to create synergy and yield more adoptions.  Thus, the intuition that it is better to saturate one market first, and then expand implicitly assumes non-submodular influence.  However, this synergy renders the problem intractable.

As we will illustrate, greedy approaches can perform poorly in these settings.
However, much of the work following Kempe et al. \cite{KempeKT03}, which
proposed the greedy algorithm, has attempted to make \emph{greedy approaches}
efficient and
scalable~\cite{chen2009efficient,ChenYZ10,lucier2015influence,cohen2014sketch,tang2014influence}.
New ideas seem necessary to design effective heuristics for non-submodular
\textsc{InfluenceMaximization}.



We observe that structural problems for networks---such as community
detection---are also, in general NP-complete, but many efficient heuristics
already exist~\cite{METIS,clausetNM2004finding}.  There are reasons to believe
that such problems are not intractable in cases likely to arise in
practice~\cite{AroraGSS12}.  This work asks whether we can design heuristics for
\textsc{InfluenceMaximization} that work well for both submodular and non-submodular cascades, and what new algorithmic techniques might efficiently find hidden synergies necessary to maximize influence.


\subsection{Contributions}
We  provide a new heuristic for solving  \textsc{InfluenceMaximization} designed to work for both submodular and non-submodular cascades.  Our algorithm
    takes as input not only a network, but a hierarchical decomposition of the
    network. 
    It then uses a dynamic programming technique to search for an
    influential seed set of nodes.  We provide the following results concerning our algorithm:
%
%
 \begin{enumerate}
  \item  We show theoretically that in certain cases, our algorithm outputs seed sets that are a factor of $\Theta(\sqrt{n})$ more influential than those of the state-of-the-art greedy algorithm, where $n$ is the number of nodes in the network.  This illustrates the intuition behind our algorithm, as well as the poor performance of greedy.
  \item  We empirically compare our algorithm to the greedy algorithm via simulations on real-world and synthetic networks for a variety of cascade models.  Our algorithm appears to do at least as well as greedy and substantially better for non-submodular cascades.  Our algorithm achieves average improvements of 7\% for submodular cascades and 55\% for non-submodular cascades, performing 266\% better in one exceptional case.
  \item   We verify the importance of network structure by showing that the quality of the hierarchical decomposition impacts the quality of our algorithm's output.
\end{enumerate}

Finally, we define a generalization of our algorithm to a ``message-passing" algorithm.  Because this algorithm is generalization, it finds seeds sets of strictly higher quality than the dynamic programming algorithm, but empirically has a longer running time.  We find that the results it returns are only marginally better than the dynamic program.  We believe that  it may be more amenable to speeding up with heuristics, but leave such studies to future work.

\subsection{Related Work}
Following the work of Kempe et al. \cite{KempeKT03}, which proposed the greedy algorithm, extensive work has constructed \emph{efficient and scalable} algorithms and heuristics \textsc{InfluenceMaximization}~\cite{chen2009efficient,ChenYZ10,nguyen2012influence,lucier2015influence,cohen2014sketch,tang2014influence}.

The heuristic algorithms presented in \cite{chen2009efficient,ChenYZ10} rely on
input parameters from the user that sacrifice accuracy for speed.  The authors
state that fine tuning the input parameters can make solving
\textsc{InfluenceMaximization} fast and accurate.
Borgs et al. provably show fast running times when the influence function is the independent cascade model \cite{borgs2012influence}.
Tang et al. extend this work to provide an algorithm that maintains
the same theoretical guarantees as the greedy algorithm presented in
\cite{KempeKT03} and is efficient in practice \cite{tang2014influence}.
Lucier et al. show how to parallelize (in a model based on Map Reduce) the subproblem of determining the influence of a particular seed \cite{lucier2015influence}.
Additional work has been done to speed up algorithms for solving
\textsc{InfluenceMaximization} by
providing techniques to efficiently compute the total influence of a seed set \cite{kimura2006tractable,cohen2014sketch}.

Leskovec et al. \cite{leskovec2007cost} consider the analogous problem of
effectively placing sensors in a network in order to effectively detect an
outbreak in the network. They present the algorithm CELF that uses a greedy
approach, but leverages the
submodularity of the cascade to reduce the amount of time it takes to evaluate
the spread of the cascade.  Moreover, CELF is built upon
by the work in \cite{goyal2011celf++,goyal2011simpath}, which present
modifications to CELF to make an even more cost effective solution to \textsc{InfluenceMaximization}.
Nguyen and Zheng present an algorithm based on belief propagation for
\textsc{InfluenceMaximization} \cite{nguyen2012influence}.  The algorithm works by systematically removing
edges until the resulting graph is a tree, and then running a belief propagation algorithm on the scaled-down
network. The authors of \cite{nguyen2012influence} show that the performance of
their algorithm is not substantially worse than that of the greedy algorithm.

In contrast to the aforementioned work, our goal is not to deliver an algorithm that is more
efficient and scalable, but rather to present an algorithm that finds higher quality seed sets.   We are unaware of any other work that claims to substantially out-perform the greedy algorithm with respect to the quality of solution.  Additionally, with the exception of~\cite{nguyen2012influence}, the prior work is based on a greedy-like approach.  Our algorithm uses a dynamic programming framework, and is fundamentally different.

Other variations of \textsc{InfluenceMaximization} have also been considered.
The works~\cite{chen2012time,liu2012time,liu2014influence} consider the problem
where the time the cascade takes to spread is constrained. Seeman and Singer study the special case where only a subset of the nodes in the network are available to be infected \cite{seeman2013adaptive}.
\textsc{InfluenceMaximization} has also been studied as a game between two different infectors \cite{Bharathi2007,GoyalK12}.

\section{Preliminaries}
\label{sec:prelim}
A real function $f$ on sets is \defn{submodular} if the marginal gain of from adding an element to a set $A$ is at least as large as the marginal gain from adding the same element to a superset $B$ of $A$.  Formally, $f$ is submodular if for all $A$, $B$, $u$ where $A \subseteq B$ we have $f(A\cup\{u\})- f(A)\geq f(B\cup\{u\})-f(B)$.

\paragraph{Cascade Model:} A cascade model is a triple $(G, F, S)$ where $G =
(V, E)$ is an unweighted graph; $F=\{f_v:\{0, 1\}^{|\Gamma(v)|} \rightarrow [0,
1]\}_{v \in V}$ is a collection of \defn{local influence functions}, where
$f_v$ takes in the set of infected neighbors of a node $v$, and produces a real
value which encodes the ``influence" of this set on $v$; and $S$ is the subset of the vertices that are initially infected.
The cascade will proceed in rounds.
In round 0, the set $S$ is infected and each of the remaining vertices is
assigned a threshold value $\theta_v \in [0, 1]$ drawn uniformly at random.
At each subsequent round, a vertex $v$ becomes infected if and only if $f_{v}(T) \geq  \theta_v$, where $T$ is the set of $v$'s infected neighbors.
We will require $f_v$ to be monotone for each $v$.

We denote the \defn{global influence function} as $\sigma(S)$ which is the \emph{expected} total number of infected vertices due to the
influence of the initial seed set $S$.

It can be shown that if $f_v$ is submodular for each $v$,then the global influence function  $\sigma$ is submodular too~\cite{MosselR10}.
Thus, we say that a model of cascade is \defn{submodular} if $f_v$ is submodular for each $v$, and is \defn{non-submodular} otherwise.


The same research that shows the $f$ usually fail to be submodular~\cite{RomeroMK11,BackstromHKL06,LeskovecAH06} shows that this submodularity fails in one particular way:  the second adopting neighbor is, on average, more influential than the first; and that after this point, each subsequent adopting neighbor's marginal influence decreases.  We call such functions \defn{2-submodular}.  Formally, $f$ is \defn{2-submodular}  if for all $A$, $B$, $A \subseteq B$, $|A|, |B| \geq 1$ and $u \not\in A, B$, we have $f(A\cup\{u\})- f(A)\geq f(B\cup\{u\})-f(B)$; and for $v \neq u$, we have $f(\{u\})- f(\emptyset)\leq f(\{u,v\})-f(\{v\})$.

Any nonzero influence function $f_v$ can be turned into a 2-submodular function by sufficiently decreasing the value of $f_v(\cdot)$ on singleton sets.

For any local influence function $f_v$, we define the \defn{$q$-deflated} version  $f^{q-defl}_v$ of $f_v$ as follows:
$$ f^{q-defl}_v(S) = \left\{ \begin{array}{cc} q \cdot f_v(S) & |S| = 1 \\ f_v(S) & \mbox{o.w.} \end{array} \right.$$

%


%

\paragraph{Specific Cascade Models:} The two popular cascade models studied in the \textsc{InfluenceMaximization}
literature are the Independent Cascade model (ICM) and the Linear
Threshold model (LTM). In the \defn{Independent Cascade model}, each newly infected node infects each
currently uninfected neighbor in the subsequent round with some fixed
probability $p$.  Thus, for all $v$, $$f^{ICM}_v(S) = 1 - (1 - p)^{|S|}.$$

In the \defn{Linear Threshold model}, each node has a threshold $\theta_v \in [0,1]$,
each of $v$'s neighbors $u$ has influence $b_{u,v}$ on $v$ such that $\sum_{u
\in \Gamma(v)} b_{u,v} \leq 1$, and $v$ becomes infected when the sum of the
influences of the infected neighbors meets or surpasses $v$'s threshold.

We define the \defn{Deflated Independent Cascade model} (DICM) which takes two parameters: $p, q \in [0, 1]$ to be the $q$-deflated version of the Independent Cascade model. 

In the \defn{S-Cascade model} (SCM) we have that $$f^{SCM}_d(S) =
\frac{\left(\frac{|S|}{2d}\right)^2}{\left(\frac{|S|}{2d}\right)^2 + \left(1 -
\frac{|S|}{d}\right)^2}.$$  This is a modified version of the Tullock Cost function~\cite{tullock1980} with power 2.

We note that the Independent Cascade model and the Linear Threshold model are
submodular, while the  $q$-Deflated Independent Cascade model (for $q < 1- p/2$)
and \emph{S}-Cascade are not.  Figure~\ref{fig:contagion_ref} illustrates the
local influence functions of the various cascade models.

\begin{figure}[h!]
  \centering
  \input{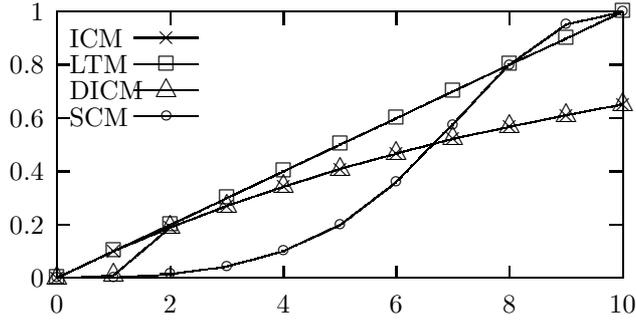}
  \caption{The local influence functions where the parameters of the ICM and
  DICM and the influence weights of LTM are all .1; and the vertex's degree in
  LTM and SCM is 10.}
  \label{fig:contagion_ref}
\end{figure}

\paragraph{Synthetic Network Model:} Most existing synthetic models fail to have meaningful asymmetry between nodes,
or significant community structure, or both.  Therefore, we design our own
synthetic network model.
\label{def:synnet}
    The \defn{directed $(d, \ell, t)$-hierarchical network model} creates a
    random network on $2^d$ nodes as follows:
    We create an edge-weighted complete binary tree of depth $d$, each leaf representing a vertex of the graph.  The weights are drawn i.i.d from a Binomial$(\ell, 1/2)$ distribution.
    Each node $v$ issues $t$ random edges, each generated via a random walk ---
    illustrated below.  Each random walk starts at $v$.  At each step in the
    walk, an outgoing edge is chosen proportional to its weight (we disallow
    exiting the node along the same edge that the walk arrived at the node).
    The walk terminates when it arrives at a leaf node.  If a terminating node is duplicated, we draw again, which keeps the graph simple.
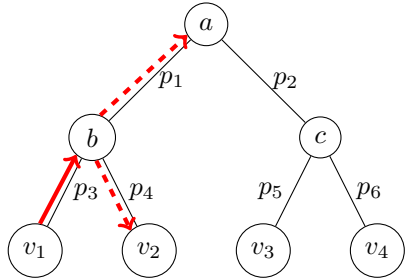
\begin{center}
\begin{tikzpicture}[level distance=1.5cm,
                    level 1/.style={sibling distance=3cm},
                    level 2/.style={sibling distance=1.5cm}]
    \node[circle,draw](a){$a$}
     child{node[circle,draw](b){$b$}
       child{node[circle,draw](v1) {$v_1$} edge from parent node[right, draw=none] {$p_3$}}
       child{node[circle,draw](v2) {$v_2$} edge from parent node[right, draw=none] {$p_4$}}
       edge from parent node[right, draw=none] {$p_1$}
     }
     child{node[circle,draw](c){$c$}
       child{node[circle,draw](v3) {$v_3$} edge from parent node[left, draw=none] {$p_5$}}
       child{node[circle,draw](v4) {$v_4$} edge from parent node[right, draw=none] {$p_6$}}
       edge from parent node[right, draw=none] {$p_2$}
     };
    \draw[red, dashed, ultra thick, ->] (b.70) -- (a.-150);
    \draw[red, dashed, ultra thick, ->] (b.-80) -- (v2.132);
    \draw[red, ultra thick, ->] (v1.80) -- (b.-132);
\end{tikzpicture}
\end{center}

As $\ell$ grows larger, this approaches the hierarchical Kleinberg model~\cite{Kleinberg2002hierarchical}.  But for moderately sized $\ell$, there is a non-trivial amount of asymmetry introduced into the graph --- some subcommunities are more influential than others.

\paragraph{Hierarchical Decomposition:} We define a \textbf{hierarchical
decomposition} of a graph $G$ to be a rooted full binary tree $T = (V_T, E_T)$
where the leaves of $T$ correspond to the vertices of $G$.  Let $r \in V_T$ be
the root of $T$. For a tree node $v\in V_T$, define $\bm{T(v)}$ to be the subset of vertices in $G$ corresponding to the leaves of the subtree rooted at $v$.   Let the \defn{\textbf{height}} of $v \in V_T$ be defined as the length of the path to $v$'s deepest descendent.

We use the recently proposed cost function of Dasgupta~\cite{Dasgupta2016} to
evaluate the quality of a hierarchical decomposition.  Let $\mathrm{lca}_T(u, v)$ be the
least common ancestor of $u, v \in V$ in the tree $T$.  Then we define
$$Cost(T) = \sum_{\{u, v\} \in E} |T(\mathrm{lca}(u, v))|$$
which sums the number of leaves in the smallest subtree containing each edge.

\paragraph{Influence Maximization:} An \textbf{\textsc{InfluenceMaximization} Instance} consists of a graph $G = (V, E)$, an influence function $\sigma$, and an integer $k$.   Given an InfluenceMaximization Instance, the goal is to find a set $S$ of $k$ nodes as to maximize $\sigma(S)$.

The \defn{greedy algorithm}~\cite{KempeKT03} for an \textsc{InfluenceMaximization} Instance start with a tentative seed set $T = \emptyset$ and for $k$ rounds, simply adds $\argmax_{v \in V} \sigma(T \cup \{v\})$ to $T$.


\section{DPIM: Dynamic Programming Influence Maximization Algorithm}
\label{sec:algorithm}
The Dynamic Programming Influence Maximization Algorithm (\algo{DPIM}), formally
specified in Algorithm \ref{algo:dpim}, takes as input a graph $G$, and corresponding hierarchical decomposition $T$,
an integer $k$, and a global influence function $\sigma(\cdot)$
and outputs a subset of vertices $S \subseteq V$ such that $|S| = k$ and $S$
is a highly influential set of seeds.
\algo{DPIM} seeks to maximize the total influence of a fixed-sized seed set $S$ by performing dynamic programming upon $T$.

For each node $v \in T$, and each $i \in \{0,1, \ldots, \min(|T(v)|,k)\}$, the algorithm stores $A[v,i]$, a choice of $i$ seeds in $T(v)$ which seeks to maximize the total influence in $G$.
Starting at the leaves of the tree, and moving up level by level until reaching the root, DPIM processes each tree node.
For each leaf node $v \in T$, we store $A[v,0] = \emptyset$ and $A[v,1] = \{v\}$.  For
each internal node $v \in T$, which has children $v_L$ and $v_R$, we set
$A[v,i]= A[v_L, j] \cup A[v_R, i-j]$  where $j \in \{0, 1, \ldots, i\}$ is
selected as to maximize $\sigma(A[ v_L,j] \cup A[ v_R, i-j])$.

\begin{algorithm}[h!]\label{alg}
\SetAlgoNoLine
\KwIn{$G = (V, E), T = (V_T, E_T), \sigma(\cdot), k$} 
\KwOut{$S \subset V$ such that $|S| = k$ 
} 
Let $A[ \cdot , \cdot ] = V_T \times [k] \to 2^V$, such that $A[v_T,j]$ stores a
choice of $j$ seeds in $T(v_T)$.  \\
Let $h$ be the height $T$.\\
\For{each height $i = 0, 1, \ldots, h$} {
  \For{each node $v \in V_T$ with height $i$} {
    \eIf{$i = 0$}{
      $A[ v, 0] = \emptyset$ \\
      $A[ v, 1] = \{v\}$
    }{
      Let $v_L, v_R$ be the left and right children of $v$,
      respectively. \\
      \For{each $i = 0, 1, \ldots, \min\{|T(v)|, k\}$}{
        $j = \argmax\limits_{ j \in \{0, 1,\ldots, i\}} \sigma(A[ v_L,j] \cup A[ v_R, i-j])$\\
        $A[v, i] = A[v_L, j] \cup A[v_R, i-j]$
      }
    }
  }
}
\Return{$A[r, k]$}
\caption{\algo{DPIM}: Dynamic Programming Influence Maximization Algorithm}
\label{algo:dpim}
\end{algorithm}

The analysis of the running time for \algo{DPIM} is straightforward.
\begin{theorem}
  Given a graph $G=(V,E)$ with $|V|=n,|E|=m$, fixed positive integers $k,r$, and a hierarchical decomposition $T$, \algo{DPIM} calls the $\sigma(\cdot)$ oracle $O(nk^2)$ times.  
\end{theorem}

\begin{proof}
  Observe that, for each node in $T$, \algo{DPIM} makes $O(k^2)$ queries to
  $\sigma(\cdot)$.    The number of
  nodes in $T$ is exactly $2n-1$. 
   Hence, the number of oracle calls  in \algo{DPIM} is $O(nk^2)$.
\end{proof}

Note that this is a factor of $k$ more than the greedy algorithm, which requires only $O(nk)$ calls to the oracle.
The execution of a single query to $\sigma(\cdot)$ can be approximated by repeatedly, $r$ times, simulating the cascade process and returning the average number of infected vertices.  This can be done in time $O(mr)$ because simulating the cascade requires at most simulating
the cascade on each edge in $G$.  However, there are often techniques to speed up the oracle access beyond simply running the cascade~\cite{borgs2012influence}, but they are beyond the scope of this work.

\subsection{Motivation:}



When \algo{DPIM} analyzes a node  $v_T \in V_T$, and the two subtrees of $v_T$ are disjoint.  If the inputs from the subtrees are optimal, then the output at of the tree will be optimal as well.  That is, \algo{DPIM} correctly decides how many nodes to allocate to each disconnected component.
Although this extreme case of disconnected subtrees may be rare, we expect the
performance to degrade gently as the number
of inter-community edges increases.
Thus, \algo{DPIM} finds the globally influential vertices by combining the knowledge it gained about each of the subtrees.  Because we expect that the community structure is important to how the cascade spreads~\cite{CentolaM07}, if the subtrees are dense ``community-like" structures, we expect our algorithm to do well.  Moreover, Community structure is hierarchical in real social networks~\cite{clausetNM2004finding}, which dynamic programming can exploit.  
However, in the same setting where you have a collection of disjoint graphs, greedy may be
suboptimal (especially if the cascade is non-submodular).  It may spread the seeds among different communities, where the optimal thing to do is saturate one community.

We illustrate this is two ways: first empirically on a random graph, and then with a theorem.

\subsubsection{Empirical Motivating Illustration}

\begin{figure*}[thp]
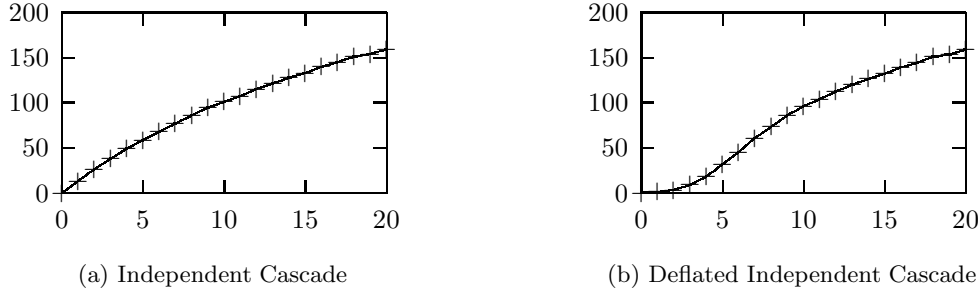

  \begin{subfigure}{.5\textwidth}
    \centering
    \input{fig/global-sub}
    \caption{Independent Cascade}
    \label{fig:sfig1}
  \end{subfigure}%
  \begin{subfigure}{.5\textwidth}
    \centering
    \input{fig/global-non-sub}
    \caption{Deflated Independent Cascade}
    \label{fig:sfig2}
  \end{subfigure}
  \caption{\label{fig:Gnp} Global influence of a random seed set on both the Independent Cascade and 0.1-Deflated Independent Cascade on a random graph with 10,000 vertices and 50,000 edges.
   The x-axis is the size of the initial seed set and the y-axis is the
total number of infections at the end to the simulation.  Notice the Global non-submodularity of Deflated Independent Cascade. }
\end{figure*}

Consider Figure~\ref{fig:Gnp} which illustrates both an Independent Cascade and 0.1-Deflated Independent Cascade on a random graph.   Notice that for this very natural graph, a small change in the local influence function creates a large change the global influence function.  This is especially true when only infecting a few seeds.  After infecting 10 seeds, there is barely any difference (101.4 versus 96.1).  However, in the Independent Cascade, the marginal impact of adding one seed, 13.3, is always greater than the average marginal impact of the first ten seeds, 10.1.  But in the Deflated Independent Cascade, the marginal impact of adding one seed, 1.3,  is not a good predictor of the average marginal impact of the first ten seeds, 9.6, which is over seven times that of adding just one node!  Adding 10 seeds in this graph can create a synergy that is not available when just considering one seed.

The greedy algorithm, which only ever considers one node at a time, may have trouble finding these synergies.   In particular, if this random graph is only part of a larger network where there are several other regions where adding one seed yields 2 infections in expectation, the greedy algorithm will never explore adding nodes to this section of the network where synergies are possible.

In contrast our  \algo{DPIM} specifically considers what happens if it adds a large number of nodes to a particular ``community-like" region.  These dense regions are exactly where we expect synergies to occur.

The greedy algorithm can make bad initial choices that are hard to mitigate later.
If the greedy algorithm adds too many seeds to graph locations that seem initially promising, but fail to provide synergies, it may not be able to catch up to the  \algo{DPIM} even if it sometimes happens upon these synergies.  In contract, \algo{DPIM} can later reallocate nodes away from a subtree that is not performing well.





%
%
\label{sec:theoretical}
\subsubsection{Theoretical Comparison in ``Worst-case" scenario}

 \label{sec:compare}
In this section, we illustrate a particular setting where Algorithm~\ref{algo:dpim} provably has considerably better performance than the greedy algorithm proposed in \cite{KempeKT03}.  This makes the aforementioned intuition rigorous.

\begin{theorem} \label{thm:worstcase}
There exist \textsc{InfluenceMaximization} instances in which the influence of the seed set output by Algorithm~\ref{algo:dpim} is a factor of $\Theta(\sqrt{N})$ larger than the influence of the seed set output by the greedy algorithm, where $N$ is the number of verices in the graph.
\end{theorem}

%

\begin{proof}
Consider the graph $G=(V,E)$ with $N=2n^2+n+2$ vertices consisting of the following three connected components:
\begin{itemize}
  \item two stars of size $n^2+1$, and
  \item one clique of size $n$.
\end{itemize}
Consider the complex contagion model with influence function $F=\{f_v\}$ such that
\begin{itemize}
  \item $f_v(S_v)=\frac1{n^2}$ if $|S_v|=1$, and
  \item $f_v(S_v)=1$ if $|S_v|\geq2$,
\end{itemize}
In the \textsc{InfluenceMaximization} instance with parameter $k=2$, we evaluate the performance of both the greedy algorithm and Algorithm~\ref{algo:dpim}.

In the greedy algorithm, the vertex with maximum marginal influence is selected.
The two obvious potential targets are the either one of the centers of the two stars, or any vertex in the clique.
If the center vertex of a star is infected, each of the remaining $n^2$ vertices will be infected with probability $\frac1{n^2}$, so the expected number of the infected vertices is $1+n^2\times\frac1{n^2}=2$.
On the other hand, if we infect any vertex in the clique, the infected vertex will not infect any other vertices in the clique with probability $\left(1-\frac1{n^2}\right)^{n-1}$, and it will infect at least one other vertex with probability $1-\left(1-\frac1{n^2}\right)^{n-1}$ in which case all the remaining vertices in the clique will be infected.
Therefore, in the case we choose a vertex in the clique, the expected number of vertices infected is
\begin{align*}
\mathbb{E}[\sigma(S)]&=1\times\left(1-\frac1{n^2}\right)^{n-1}+n\times\left(1-\left(1-\frac1{n^2}\right)^{n-1}\right)\\
&=1+(n-1)\left(1-\left(1-\frac1{n^2}\right)^{n-1}\right)\\
&=1+(n-1)\frac1{n^2}\left(\sum_{i=0}^{n-2}\left(1-\frac1{n^2}\right)^i \right)\\
&<1+(n-1)\frac1{n^2}\cdot(n-1)\\
&<2.
\end{align*}
Thus, infecting the center of a star has higher marginal influence, so the greedy algorithm will choose the two centers of the two stars to infect, and an expected total number of 4 vertices will be infected.

On the other hand, assuming that the hierarchical decomposition $T$ contains a subtree $T(v)$ consisting only of the clique (which it should as it is disconnected from the rest of the graph):
Algorithm~\ref{algo:dpim} will choose 2 vertices from the clique causing a total number of $n$ infected vertices, which is also the optimal strategy.
To see this, note that Algorithm~\ref{algo:dpim} will process $v$, and at this point will compute $\sigma(S)$ where $S$ contains two nodes from the clique.   By the nature of Algorithm~\ref{algo:dpim}, this solution, or a better one, will be considered (inductively) for each vertex that has $v$ as a descendent.  Thus the output seed set is at least as influential as placing two seeds in the clique, which has influence $n = \Theta(\sqrt{N})$.
Consequently, the performance of Algorithm~\ref{algo:dpim} is better than the performance of the greedy algorithm by a factor of
$$\frac{n}4=\Theta(n)=\Theta(\sqrt N).$$
\end{proof}
Notice that the structure of the graph in the example above is related to some
of the common structures found in social networks, which contain many different
communities with different internal structures.

\section{Experimental Results}
\label{sec:emp_results}
\subsection{Experimental Setup}
  \label{ssec:exp_setup}

We execute \algo{DPIM} and the greedy algorithm from \cite{KempeKT03} on a variety of
networks and cascades to test the relative quality of solutions.


\subsubsection{Cascade Models}
We adopt the two common submodular cascade models from the literature: the linear threshold model and the independent cascade model, defined in Section~\ref{sec:prelim}, and two non-submodular cascades:

\begin{enumerate}[I)]
\item \underline{\emph{Independent Cascade} (IC):} We uniformly assign the
probability $p = 1\%$, thus $v$ with $\ell$ infected neighbors is infected with
probability $1 - (0.99)^{\ell}$.
\item \underline{\emph{Linear Threshold} (LT):} For each node $v$, we assign each of
$u \in \Gamma(v)$ to have $1/|\Gamma(v)|$ influence on $v$.
\item \underline{\emph{Deflated Independent Cascade} (DIC):} We uniformly assign the
probability $p = 1\%$, thus $v$ with $\ell = 1$ infected neighbors is infected with probabilty
$0.001$ and with $\ell \geq 2$ infected neighbors is infected with probability $1 - (0.99)^{\ell}$.
\item \underline{\emph{S-Cascade model} (SCM):} The influence
on any given node $v$ is
\[\frac{(x/2)^2}{(x/2)^2 + (1 - x)^2},\]
where $x$ is the fraction of $v$'s neighbors that are infected.
\end{enumerate}

Both of these algorithms require access to an oracle for $\sigma(\cdot)$, which is also required to evaluate the effectiveness of the algorithms.  To implement this oracle, we simulate the cascade 100 times, resampling the randomness for the cascade each time (using pseudorandomness from the standard C++ library) and return the average number of infections.

\subsubsection{Networks}

We use two real-world networks (from \cite{snapnets}) and two synthetic networks, summarized in Table \ref{table:nets}.
In the arXiv collaboration network (ca-GrQc), the vertices are authors of e-print scientific articles and edges represent coauthorship relations.
The ego-Facebook network is largest such network provided by \cite{snapnets}.  This network denotes the facebook friendship ties from a single person's (ego's) set of friends.  The ego vertex has been removed.
Furthermore, we generate two synthetic networks by first sampling from directed $(d, \ell,
t)$-hierarchical network model using parameters $(10, 50, 50)$ and $(11, 50,
50)$ (synthetic-1 \& synthetic-2, resp.), and then making the graph simple and
undirected in the natural way.  
  \begin{table}[htb]
    \centering
    \begin{tabular}{|l|c|c|}
    \hline
      \textbf{Name} & \textbf{Nodes} & \textbf{Edges} \\
    \hline
      synthetic-1 & 1,024 & 51,200 \\
    \hline
      synthetic-2 & 2,048 & 102,400 \\
    \hline
      ca-GrQc & 5,276 & 28,827 \\
    \hline
      ego-Facebook & 1,034 & 53,498 \\
    \hline
    \end{tabular}
    \caption{Networks used to evaluate the effectiveness of our algorithm.}
    \label{table:nets}
  \end{table}
\subsubsection{Algorithms for Hierarchical Decomposition}
  Lastly, in order to evaluate our algorithm, we present 4 algorithms for
  generating a hierarchical decomposition of any network. The algorithms
  we used in our simulations are implemented as follows:

\begin{enumerate}[I)]
\item \underline{\emph{Random Pair}:} Each node starts in
its own partition, and partitions are joined randomly until all of the nodes are
contained in one partition.
\item \underline{\emph{Random Edge}:} Each node starts in
its own partition, and partitions are joined by contracting a random edge between
partitions.  If no edges remain between the partitions, partitions are merged randomly until all
of the nodes are contained in one partition.
\item \underline{\emph{Jaccard Similarity}:} Each node starts in
its own partition, and pairs of partitions $(A,B)$, for $A,B \subset V$ are joined based on
which pair maximizes
\[\frac{|\Gamma(A) \cap \Gamma(B)|}{|\Gamma(A) \cup \Gamma(B)|},\]
where $\Gamma(X \subset V) = \bigcup_{v \in X} \Gamma(v)$.
\item \underline{\emph{METIS-based}:} The whole network starts as one partition; using METIS~\cite{METIS}, partitions are
recursively divided into two partitions  until each partition contains only a single node.
\end{enumerate}

\subsection{Algorithm Evaluation}
\label{ssec:algo_eval}

\subsubsection{Performance of DPIM}
  \begin{figure*}[htbp]
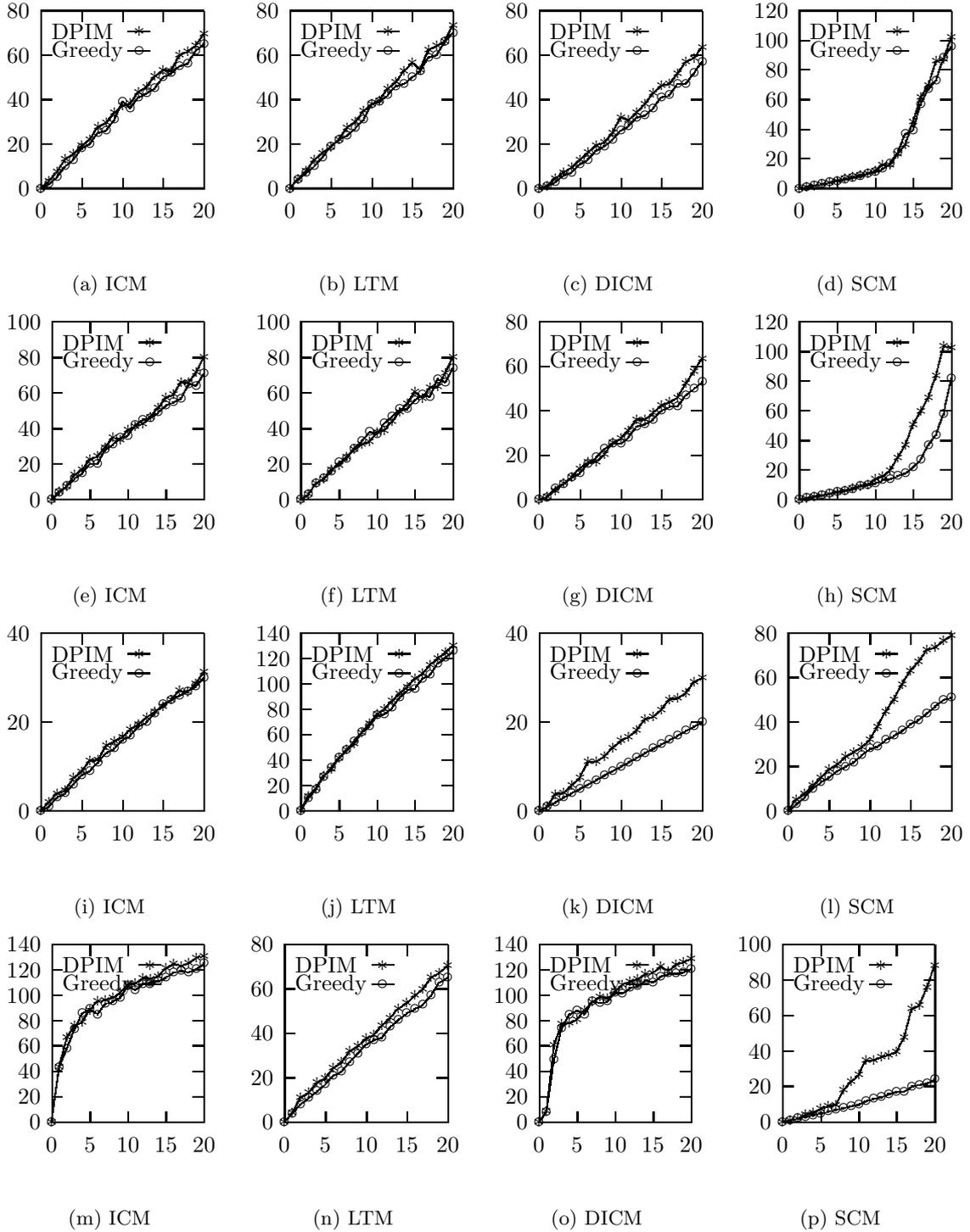

    \begin{subfigure}{0.24\textwidth}
      \centering
      \input{fig/IndependentCascade.syn10d50c50a-txt}
      \caption{ICM}
    \end{subfigure}
    \begin{subfigure}{0.24\textwidth}
      \centering
      \input{fig/LinearThreshold.syn10d50c50a-txt}
      \caption{LTM}
    \end{subfigure}
    \begin{subfigure}{0.24\textwidth}
      \centering
      \input{fig/DeflatedIndependentCascade.syn10d50c50a-txt}
      \caption{DICM}
    \end{subfigure}
    \begin{subfigure}{0.24\textwidth}
      \centering
      \input{fig/GoyalContagion.syn10d50c50a-txt}
      \caption{SCM}
    \end{subfigure}
    \begin{subfigure}{0.24\textwidth}
      \centering
      \input{fig/IndependentCascade.syn11d50c50a-txt}
      \caption{ICM}
    \end{subfigure}
    \begin{subfigure}{0.24\textwidth}
      \centering
      \input{fig/LinearThreshold.syn11d50c50a-txt}
      \caption{LTM}
    \end{subfigure}
    \begin{subfigure}{0.24\textwidth}
      \centering
      \input{fig/DeflatedIndependentCascade.syn11d50c50a-txt}
      \caption{DICM}
    \end{subfigure}
    \begin{subfigure}{0.24\textwidth}
      \centering
      \input{fig/GoyalContagion.syn11d50c50a-txt}
      \caption{SCM}
    \end{subfigure}
    \begin{subfigure}{0.24\textwidth}
      \centering
      \input{fig/IndependentCascade.ca-GrQc-txt}
      \caption{ICM}
    \end{subfigure}
    \begin{subfigure}{0.24\textwidth}
      \centering
      \input{fig/LinearThreshold.ca-GrQc-txt}
      \caption{LTM}
    \end{subfigure}
    \begin{subfigure}{0.24\textwidth}
      \centering
      \input{fig/DeflatedIndependentCascade.ca-GrQc-txt}
      \caption{DICM}
    \end{subfigure}
    \begin{subfigure}{0.24\textwidth}
      \centering
      \input{fig/GoyalContagion.ca-GrQc-txt}
      \caption{SCM}
    \end{subfigure}
    \begin{subfigure}{0.24\textwidth}
      \centering
      \input{fig/IndependentCascade.107-edges-txt}
      \caption{ICM}
    \end{subfigure}
    \begin{subfigure}{0.24\textwidth}
      \centering
      \input{fig/LinearThreshold.107-edges-txt}
      \caption{LTM}
    \end{subfigure}
    \begin{subfigure}{0.24\textwidth}
      \centering
      \input{fig/DeflatedIndependentCascade.107-edges-txt}
      \caption{DICM}
    \end{subfigure}
    \begin{subfigure}{0.24\textwidth}
      \centering
      \input{fig/GoyalContagion.107-edges-txt}
      \caption{SCM}
    \end{subfigure}

    \caption{\textbf{Comparison of performance: \algo{DPIM} vs. Greedy}. The
    rows from top to bottom correspond to synthetic-1, synthetic-2, ca-GrQc,
    and ego-Facebook, respectively. For
    each plot, the x-axis is $k$ and the y-axis is the number of total infections
    at the end of the cascade.}
    \label{fig:eval_s2}
  \end{figure*}

 The results of the simulations we ran are shown in Figure \ref{fig:eval_s2}.  For each
 execution of \algo{DPIM}, we used the METIS-based hierarchical decomposition
 algorithm to construct a hierarchical decomposition of the network.

Considering cascades across all four networks with seed set size 20, \algo{DPIM} increases influence on average by 8\% for ICM,  6\% for LTM, 22\%  for DICM, and  88\%  for SCM.
%
%
%
 Surprisingly, \algo{DPIM} performs marginally better than the greedy algorithm
 even for submodular cascades.  As predicted, when the cascade is non-submodular, \algo{DPIM} outperforms the greedy algorithm by a significant amount.  However, gains were more impressive for
 synthetic-2, ca-GrQc, and ego-Facebook --- including an 266\% increase in
 influence for the SCM cascade on the ego-Facebook network --- than for
 synthetic-1, where we see only marginal improvement even when the cascade is
 non-submodular.  Table \ref{table:eval} contains the approximated expected
 total influence values for each simulations rounded to the nearest integer.

 In addition, we present the running times of each simulation (both \algo{DPIM}
 and Greedy) in Table \ref{table:times}. Theoretically, the
 greedy algorithm queries $\sigma(\cdot)$ $O(nk)$ times, and \algo{DPIM} queries $\sigma(\cdot)$
 $O(nk^2)$ times.  Despite this,
 the empirical results show that $\algo{DPIM}$ is roughly a small constant factor
 slower than the greedy algorithm, and occasionally much faster.

 \begin{table*}[htbp]
 \label{multiprogram}
 \centering
 \begin{tabular}{l|c|c|c|c|c|c|c|c|}
  & \multicolumn{2}{c|}{\textbf{synthetic-1}}
  & \multicolumn{2}{c|}{\textbf{synthetic-2}}
  & \multicolumn{2}{c|}{\textbf{ca-GrQc}}
  & \multicolumn{2}{c|}{\textbf{ego-Facebook}} \\
   \cline{2-9}
& Greedy & \algo{DPIM}  & Greedy & \algo{DPIM}
& Greedy & \algo{DPIM} & Greedy & \algo{DPIM} \\ \hline
  \hline
    \textbf{ICM} &  23,720 & 90,945  & 49,885 & 162,892 & 39,641 & 350,006 &
    275,516 & 758,843  \\
  \hline
    \textbf{LTM} & 24,121 & 92,786 &  52,018 & 165,369 & 25,626 & 56,156 & 308,637
    & 83,584 \\
  \hline
    \textbf{DICM} & 18,536 & 69,993 & 42,176 & 134,744 & 6,108 & 26,672 & 266,859 & 641,962 \\
  \hline
    \textbf{SCM} & 13,676 & 36,765 & 23,749 & 69,737 & 10,274 & 46,196 & 38,279 & 49,892 \\
  \hline
   \end{tabular}
 \caption{Running times in seconds of each of the simulations.}
 \label{table:times}
 \end{table*}

 \begin{table*}[htbp]
 \label{multiprogram}
 \centering
 \begin{tabular}{l|c|c|c|c|c|c|c|c|}
  & \multicolumn{2}{c|}{\textbf{synthetic-1}}
  & \multicolumn{2}{c|}{\textbf{synthetic-2}}
  & \multicolumn{2}{c|}{\textbf{ca-GrQc}}
  & \multicolumn{2}{c|}{\textbf{ego-Facebook}} \\
   \cline{2-9}
& Greedy & \algo{DPIM}  & Greedy & \algo{DPIM}
& Greedy & \algo{DPIM} & Greedy & \algo{DPIM} \\ \hline
    \textbf{ICM} & 65 & 70 & 71 & 80 & 30 & 31 & 125 & 131 \\
  \hline
    \textbf{LTM} & 70 & 74 & 74 & 81 & 126 & 130 & 65 & 70 \\
  \hline
    \textbf{DICM} & 57 & 64 & 53 & 64 & 20 & 30 & 121 & 129 \\
  \hline
    \textbf{SCM} & 96 & 102 & 82 & 103 & 51 & 79 & 24 & 88 \\
  \hline
   \end{tabular}
 \caption{Expected total influence of the final seed sets of size 20 chosen by
 both algorithms for each of the simulations (rounded to the nearest integer).}
 \label{table:eval}
 \end{table*}

\subsubsection{Comparison of Hierarchical Decomposition Algorithms}
For each network, we tested how \algo{DPIM} performed over the various
hierarchical decomposition algorithms with SCM as the cascade model.
In addition, we recorded the cost of each
hierarchical decomposition (see Section~\ref{sec:prelim} for details).  The details of the performance of our algorithm are
shown in Figure \ref{fig:hd_eval} and the cost of each of the hierarchical
decompositions is noted in Table \ref{table:costs}. Observe that \algo{DPIM} performs
significantly better with the hierarchical
decompositions that have lower costs.

Despite performing better when the hierarchical decomposition better
represents the community structure of the network, there seem to be diminishing
returns beyond a certain cost for each network. The METIS-based
algorithm and the Jaccard Similarity algorithm produce hierarchical
decompositions that differ in cost by a relatively large amount, but the
difference in performance of our algorithm between the two different hierarchical
decompositions is small.

Since a lower cost valuation implies that the hierarchical
decomposition better represents the structure of the network, this result provides
evidence that our algorithm is leveraging the community
structure provided by the hierarchical decomposition.

  \begin{table*}[h!]
    \centering
    \begin{tabular}{l|c|c|c|c|}
       & \textbf{synthetic-1} & \textbf{synthetic-2} &
      \textbf{ca-GrQc} & \textbf{ego-Facebook}\\
    \hline
      \textbf{Random Pair} & 35,063,945 & 139,590,514, & 99,792,540 & 36,387,204 \\
    \hline
      \textbf{Random Edge} & 34,084,502 & 133,759,197 & 57,800,479 & 31,092,978 \\
    \hline
      \textbf{Jaccard Similarity} & 22,508,941 & 88,531,982 & 27,695,152 & 13,455,920 \\
    \hline
      \textbf{METIS-based} & 5,267,974 & 12,185,692 & 18,425,481 & 8,968,620 \\
    \hline
    \end{tabular}
    \caption{Cost of the hierarchical decompositions produced by each algorithm
    for each network.}
    \label{table:costs}
  \end{table*}

  \begin{figure*}[h!]
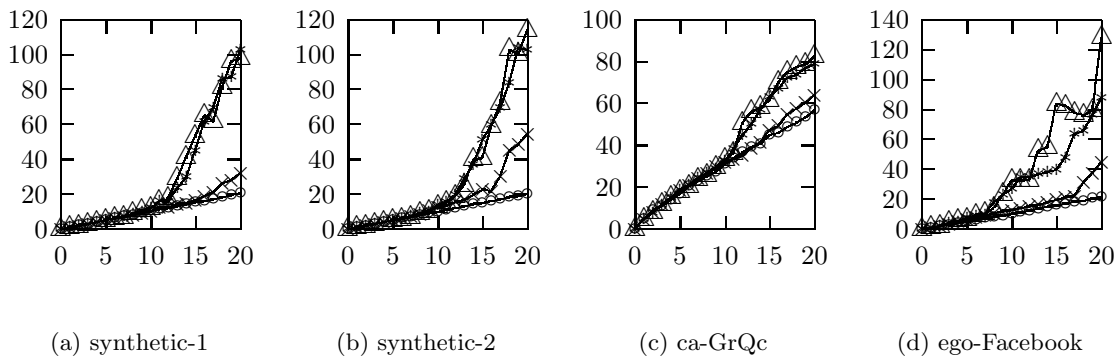

    \begin{subfigure}{0.24\textwidth}
      \centering
      \input{fig/GoyalContagion.syn10d50c50a-txt.HD}
      \caption{synthetic-1}
    \end{subfigure}
    \begin{subfigure}{0.24\textwidth}
      \centering
      \input{fig/GoyalContagion.syn11d50c50a-txt.HD}
      \caption{synthetic-2}
    \end{subfigure}
    \begin{subfigure}{0.24\textwidth}
      \centering
      \input{fig/GoyalContagion.ca-GrQc-txt.HD}
      \caption{ca-GrQc}
    \end{subfigure}
    \begin{subfigure}{0.24\textwidth}
      \centering
      \input{fig/GoyalContagion.107-edges-txt.HD}
      \caption{ego-Facebook}
    \end{subfigure}

    \caption{\textbf{Comparison of Hierarchical Decomposition Algorithms}. For
    each plot, the x-axis is $k$ and the y-axis is the number of total infections
    at the end of the cascade. In each plot,  the
    circle is Random Pair, the cross is Random Edge, the triangle is Jaccard
    Similarity, and the asterisk is METIS-based. Each seed set was choosen by \algo{DPIM} with each of the
    respective hierarchical decomposition algorithm as input, and the cascade
    used was SCM.}
    \label{fig:hd_eval}
  \end{figure*}

\section{Message Passing Algorithm}
\label{sec:message-passing}
In this section, we describe a Message Passing Algorithm (\algo{MPA}), formally presented in Algorithm \ref{algo:mpa}, which is a generalization of  \algo{DPIM}.  This algorithm provides a new perspective on  \algo{DPIM}, and can be made to perform better at the cost of additional run time.

\paragraph{Directional Subtrees and Their Recursive Decomposition:} First, we must present notation and ideas that we will use to formally describe
\algo{MPA}.   For a node $v_T \in V_T$, we define $L(v_T), R(v_T),$ and $U(v_T)$ to be the left child, right child, and parent (up), respectively.  We use $D$ as a placeholder for an element in $\{L, R, U\}$ (left, right, up directions), so that $D(v_T)$ refers to either $L(v_T), R(v_T),$ or $U(v_T)$.   Let $\mathbb{D} =
\{(L,R), (L,U), (R,U) \}$.

Furthermore, for node, $v_T \in
V_T$ and a direction  $D \in \{L, R, U\}$ we define the backward direction $B^{(v_T)}_D \in \{L, R, U\}$ to be the direction one takes from $D(v_T)$ in order to return to $v_T$ so that $B^{(v_T)}_D(D(v_T)) = v_T$.  
We naturally define $\neg
B^{(v_T)}_D \in \mathbb{D}$ such that neither of the
coordinates of $\neg B^{(v_T)}_D$ are
$B^{(v_T)}_D$.

Lastly, for each node $v_T \in V_T$ we define three \emph{directional subtrees}:
\emph{left} --- $T_L(v_T) = T(L(v_T))$, \emph{right} --- $T_R(v_T) = T(R(v_T))$,
and \emph{up} --- $T_U(v_T) = V \setminus (T(v_T))$; which partition the
vertices of $G$ because $T(v_T) = T_L(v_T) \cup T(R(v_T))$.  A key observation
is that each of these three subtrees, can be decomposed into smaller directions
subtrees of $L(v_T)$, $R(v_T)$, and $U(v_T)$ respectively.  For example, for $D
\in \{L, R\}$ we have that $D(v_T) = T_L(D(v_T)) \cup T_R(D(v_T))$.  And while
$v_T$ is not the root of $V_T$, we have $T_U(v_T) = \bigcup_{D \in \neg B^{v_T}_U} T_D(U(v_T))$.

\paragraph{Data Structure:} \algo{MPA} will sequentially, locally update a data-structure, which we describe
now. Let $A: V_T \times \mathbb{D} \times [k] \rightarrow [k] \times [k]$
where both output coordinates of $A(v_T, (D_1, D_2), \ell)$ are non-negative integers that sum to at most $\ell$.
Intuitively, $A(v_T, (D_1, D_2), \ell)$ answers:  how
should  $\ell$ seeds should be split between $T_{D_1}(v_T)$ and
$T_{D_2}(v_T)$? The answer is that we should put  $A(v_T,
(D_1, D_2), \ell)_i$ seeds in $T_{D_i}(v_T)$ for
each $i \in \{1, 2\}$. For the root vertex $r \in V_T$, which has no parent, we insist that  $A(r,
(D,U), \ell)_2 = 0$ for $D \in \{L, R \}$.


\paragraph{Obtaining a Seed Set:}Once we have the data-structure $A$, we define a function that answers the more natural question:
Let $C: V_T \times \mathbb{D} \times [k] \rightarrow 2^V$ where
$C(v_T,  (D_1,D_2), \ell)$ is a subset of $
T_{D_1}(v_T) \cup T_{D_2}(v_T)$ of at most $\ell$
vertices.  Intuitively, the question that $C$ answers is: where should we allocate
$\ell$ seed vertices in $T_{D_1}(v_T) \cup T_{D_2}(v_T)$?
Given an instance of $A$, $C(v_T, (D_1, D_2), \ell)$  recursively follows the advice of $A$ in
placing seeds in $T_{D_1}(v_T)$ and $T_{D_2}(v_T)$.  Intuitively, we can do this using $A$ and the recursive tree decomposition.  Formally, we
divide the definition of $C(v_T, (D_1, D_2), \ell)$ into two
cases:

\begin{enumerate}[I)]
\item $(D_1, D_2) = (L, R)$: If $v_T$ is a leaf node, then
$$C(v_T, (L,R), \ell) = \left\{ \begin{array}{cc} T(v_T) &  \ell \geq 1 \\
\emptyset & \ell = 0 \end{array}.\right.$$
Otherwise,
\begin{align*}
C(v_T, (L,R), \ell) = & C(L(v_T), (L,R), A[v_T, (L,R), \ell]_1) \; \cup \\ &C(R(v_T),
(L,R), A[v_T, (L,R), \ell]_2).
\end{align*}
\item $(D_1, D_2) = (D, U)$  for $D \in \{L, R \}$: If $v_T = r$, the root of
$T$, then
$$C(v_T, (D,U), \ell) = C(D(v_T), (L,R), \ell).$$
If $v_T$ is a leaf node, then it is undefined (or just $\emptyset$).
Otherwise,
\begin{align*}
C(v_T, (D,U), l) = & C(D(v_T), (L, R), A(v_T, (D, U), l)_1) \; \cup \\
& C(U(v_T), \neg B^{(v_T)}_D, A(v_T, (D, U), l)_2).
\end{align*}

\end{enumerate}

%
%

%

\noindent Note that $C(r, (L,R), k)$ denotes where to place $k$ initial seeds over all of $V$.

\paragraph{Local Update Subroutine:}

Now that we have defined $A$ and shown how it defines initial seeds sets,  we
will define a way to locally update $A$ to improve the allocation.  Intuitively, we update  $A(v_T, (D_1, D_2), \ell)$ by trying all possible outputs whose
coordinates sum to $\ell$ and returning the best.  Let $D_3$ be the sole
element of ${\{ L, R, U\} \setminus\{D_1, D_2}\}$.  We
will first allocate $k - \ell$ vertices to $T_{D_3} (v_T)$ according to $C$.  Note that this tree can be decomposed into two directional subtrees of $D_3(v_T)$ upon which $C$ is defined.
Next, we try the different divisions between $T_{D_1}(v_T)$ and
$T_{D_2}(v_T)$ for the remaining $\ell$ vertices. This procedure is
formally defined in Algorithm \ref{algo:update}.
%
%

\begin{algorithm}[h]\label{alg}
\SetAlgoNoLine
\KwIn{$I = (G = (V, E), T = (V_T, E_T), \sigma(\cdot), k), A, v_T,
(D_1, D_2) \in \mathbb{D}$} 
\KwOut{An updated instance of $A$} 
Let $D_3$ be the sole element of $\{ L, R, U\} \setminus\{D_1,
D_2\}$.\\
%

 \For{each $i = 0, 1, \ldots, k$}{
        $j = \argmax\limits_{ j \in \{0, 1,\ldots, i\}}
        \sigma\left(
        C\left(D_1(v_T),\neg R^{(v_T)}_{D_1}, j\right) \cup
        C\left(D_2(v_T),\neg R^{(v_T)}_{D_2},
        i-j\right) \cup C\left(D_3(v_T), \neg R^{(v_T)}_{D_3},
        k-i\right) \right)$\\
        $A(v_T, (D_1, D_2), i) := (j, i -j)$ \\
       }
  \Return{A}
\caption{\algo{update}: Local Update Subroutine}
\label{algo:update}
\end{algorithm}

\paragraph{Generalizing \algo{DPIM}:} Given that we have a local update, we
need only define: (a) an initial configuration of $A$, (b) a sequence of
updates, and (c) a terminating condition (if the sequence is infinite).
For example, our  \algo{DPIM} is a special case where every entry in $A$ is $(0, 0)$
initially, and then we update $A(v_T, (L,R), \cdot)$ by going up the tree.

\begin{algorithm}[h]\label{alg}
\SetAlgoNoLine
\KwIn{$I = (G = (V, E), T = (V_T, E_T), \sigma(\cdot), k)$} 
\KwOut{$S$, a set of $k$ nodes of $V$}
Let $A(\cdot) = (0, 0)$ \\
\For{each height $i = 1, 2, \ldots h$} {
  \For{each node $v_T \in V_T$ with height $i$} {
      $\algo{update}(I, A, v_T, (L, R));$\\
  }
}
Let $S, S' := C(r_T, (L, R), k)$;\\
\Do{$\sigma(S') > \sigma(S)$}{
  S := S';\\
  \For{each height $i = h-1, \ldots, 1$} {
    \For{each node $v_T \in V_T$ with height $i$} {
      $\algo{update}(I, A, v_T, (L, U))$;\\
      $\algo{update}(I, A, v_T, (R, U))$;
    }
  }
  \For{each height $i = 1, 2, \ldots h$} {
    \For{each node $v_T \in V_T$ with height $i$} {
      $\algo{update}(I, A, v_T, (L, R));$\\
    }
  }
  $S' = C(r_T, (L, R), k)$;
}
\Return{$C(r_T, (L, R), k)$}
\caption{\algo{MPA}: Message Passing Algorithm}
\label{algo:mpa}
\end{algorithm}

In Algorithm~\ref{algo:mpa}, we give a possible schedule where we go up
and down the tree until no additional improvements are found. Given a graph, decomposition,
influence function, $k$, and current allocation $A$, we run \algo{DPIM} to
initialize $A$.  Then, we evaluate each node on each level of $T$ down and up
the tree, finishing at the root node, checking whether there has been any
improvements in the influence of the seed sets, and repeating until there are no
more improvements. Note that this algorithm must terminate, since requires improvement in each round and there are only a finite number
of possible configurations.  

\paragraph{Intuition:}  When \algo{DPIM} sets $A(v_T, (L, R), \ell)$ for $\ell <
k$, it does not know where the other seeds (outside $T(v_T)$ will be placed.
Thus, a better decision may be available with this added information.  How
should an algorithm decide how to place elements in $V \setminus T(v_T)$?  If we
only know $A(\cdot, (L, R), \cdot)$, this does not appear to be enough
information.  At first blush, we must know how many seeds to allocate to the
other subtree of $U(v_T)$, and how many to send further up the tree.  This is
exactly what $A(U(v_T), \neg B^{v_T}_U, \cdot)$ tells us!  Of course, all the
$A(v_T, \cdot, \cdot)$'s seem interdependent.  Thus, we create a local update algorithm to gradually refine each of them.

\subsection{Message Passing Algorithm Evaluation}

We evaluate \algo{MPA} by comparing it's effectiveness against \algo{DPIM} for
all four cascades.  Due to the \emph{significant time complexity} of \algo{MPA}, we were
only able to simulate \algo{MPA} on two small
networks --- we sampled from the directed $(d, l, t)$-hierarchical network model
using parameters $(8, 15, 15)$ and have a smaller ego-network from Facebook that
has 150 nodes and 3,386 edges.  The results from these simulations can be seen
in Figures \ref{fig:mpa_eval} and \ref{fig:mpa_eval_rw}.  Notice that we have
plots for $k = 5, 10, 15, 20$, since \algo{MPA} optimizes for the
$k$ exactly and does not guarentee that is is finding the best seeds for seeds
sets of size $1,\ldots,k-1$. This way we get to see how effective \algo{MPA} is
at multiple sizes of seed sets.

As can be seen in Figures \ref{fig:mpa_eval} and \ref{fig:mpa_eval_rw}, the improvement that \algo{MPA} has
over \algo{DPIM} is marginal for the ICM, LTM, and DICM cascades (averageing
improvements of 3\%, 1\%, 3\%, respectively).  On the other
hand, \algo{MPA} finds a seed set that has significantly more influence when the
cascade is SCM (with an average improvement of 16\%).  It seems like for an algorithm to find a high-quality seed set
for the SCM cascade, the algorithm must be able to find intercommunity synergies
which are the only way that SCM will propogate throughout networks.


  \begin{figure*}[hbtp]
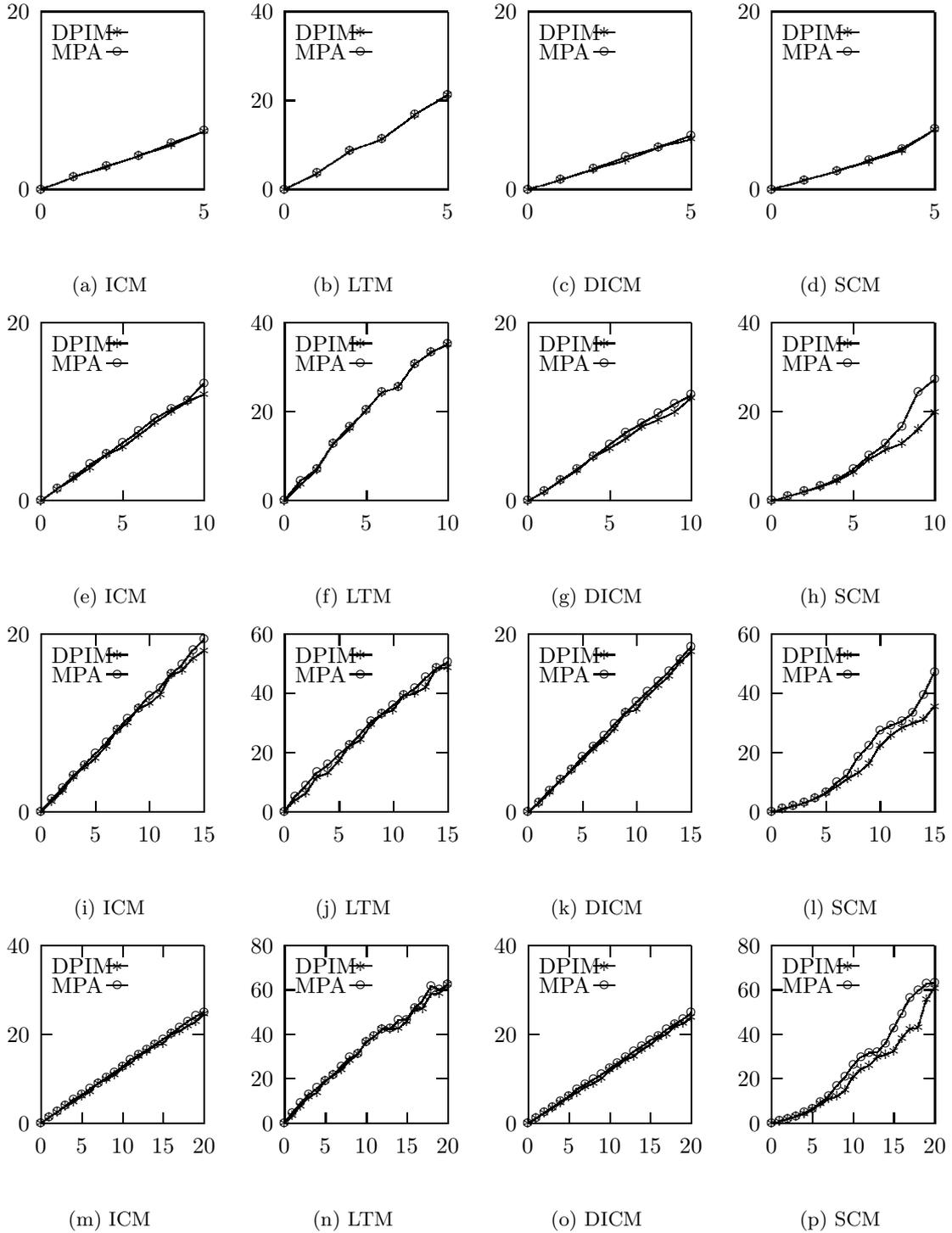

    \begin{subfigure}{0.24\textwidth}
      \centering
      \input{fig/mpa.IndependentCascade.syn8d15c15a-txt.5}
      \caption{ICM}
    \end{subfigure}
    \begin{subfigure}{0.24\textwidth}
      \centering
      \input{fig/mpa.LinearThreshold.syn8d15c15a-txt.5}
      \caption{LTM}
    \end{subfigure}
    \begin{subfigure}{0.24\textwidth}
      \centering
      \input{fig/mpa.DeflatedIndependentCascade.syn8d15c15a-txt.5}
      \caption{DICM}
    \end{subfigure}
    \begin{subfigure}{0.24\textwidth}
      \centering
      \input{fig/mpa.GoyalContagion.syn8d15c15a-txt.5}
      \caption{SCM}
    \end{subfigure}

    \begin{subfigure}{0.24\textwidth}
      \centering
      \input{fig/mpa.IndependentCascade.syn8d15c15a-txt.10}
      \caption{ICM}
    \end{subfigure}
    \begin{subfigure}{0.24\textwidth}
      \centering
      \input{fig/mpa.LinearThreshold.syn8d15c15a-txt.10}
      \caption{LTM}
    \end{subfigure}
    \begin{subfigure}{0.24\textwidth}
      \centering
      \input{fig/mpa.DeflatedIndependentCascade.syn8d15c15a-txt.10}
      \caption{DICM}
    \end{subfigure}
    \begin{subfigure}{0.24\textwidth}
      \centering
      \input{fig/mpa.GoyalContagion.syn8d15c15a-txt.10}
      \caption{SCM}
    \end{subfigure}

    \begin{subfigure}{0.24\textwidth}
      \centering
      \input{fig/mpa.IndependentCascade.syn8d15c15a-txt.15}
      \caption{ICM}
    \end{subfigure}
    \begin{subfigure}{0.24\textwidth}
      \centering
      \input{fig/mpa.LinearThreshold.syn8d15c15a-txt.15}
      \caption{LTM}
    \end{subfigure}
    \begin{subfigure}{0.24\textwidth}
      \centering
      \input{fig/mpa.DeflatedIndependentCascade.syn8d15c15a-txt.15}
      \caption{DICM}
    \end{subfigure}
    \begin{subfigure}{0.24\textwidth}
      \centering
      \input{fig/mpa.GoyalContagion.syn8d15c15a-txt.15}
      \caption{SCM}
    \end{subfigure}

    \begin{subfigure}{0.24\textwidth}
      \centering
      \input{fig/mpa.IndependentCascade.syn8d15c15a-txt.20}
      \caption{ICM}
    \end{subfigure}
    \begin{subfigure}{0.24\textwidth}
      \centering
      \input{fig/mpa.LinearThreshold.syn8d15c15a-txt.20}
      \caption{LTM}
    \end{subfigure}
    \begin{subfigure}{0.24\textwidth}
      \centering
      \input{fig/mpa.DeflatedIndependentCascade.syn8d15c15a-txt.20}
      \caption{DICM}
    \end{subfigure}
    \begin{subfigure}{0.24\textwidth}
      \centering
      \input{fig/mpa.GoyalContagion.syn8d15c15a-txt.20}
      \caption{SCM}
    \end{subfigure}

    \caption{\textbf{Evaluation of \algo{MPA} on a synthetic network.} We evaluate the effectiveness of
    \algo{MPA}
    by comparing the results to \algo{DPIM}.  The network was a sample from the
    directed $(d,l,t)$- hierarchical network model using parameters (8,15,15).
    For each plot, the x-axis is $k$ and the y-axis is the expected total
    influence from the seed set chosen by each respective algorithm.}
    \label{fig:mpa_eval}
  \end{figure*}

  \begin{figure*}[hbtp]
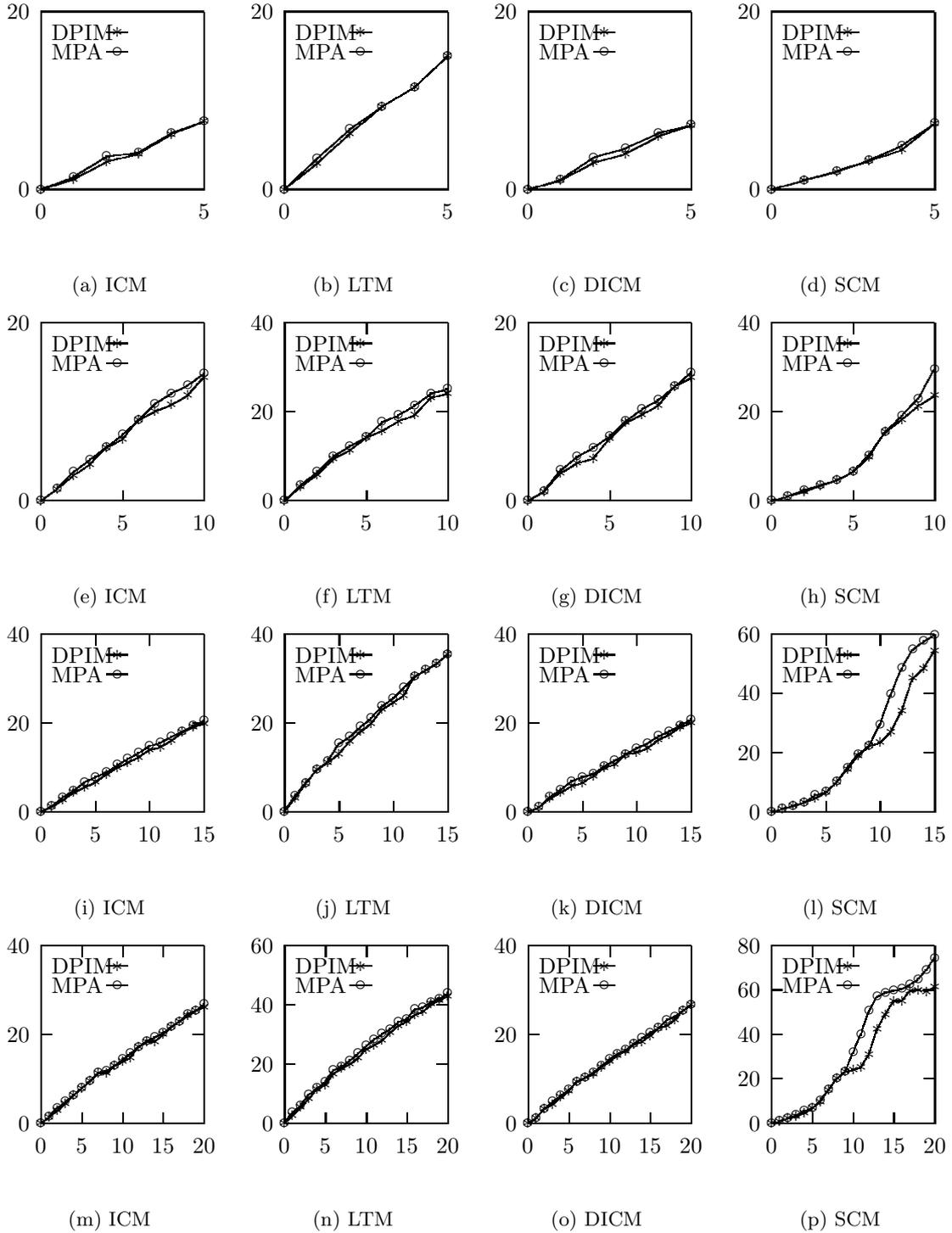

    \begin{subfigure}{0.24\textwidth}
      \centering
      \input{fig/mpa.IndependentCascade.414-edges-txt.5}
      \caption{ICM}
    \end{subfigure}
    \begin{subfigure}{0.24\textwidth}
      \centering
      \input{fig/mpa.LinearThreshold.414-edges-txt.5}
      \caption{LTM}
    \end{subfigure}
    \begin{subfigure}{0.24\textwidth}
      \centering
      \input{fig/mpa.DeflatedIndependentCascade.414-edges-txt.5}
      \caption{DICM}
    \end{subfigure}
    \begin{subfigure}{0.24\textwidth}
      \centering
      \input{fig/mpa.GoyalContagion.414-edges-txt.5}
      \caption{SCM}
    \end{subfigure}

    \begin{subfigure}{0.24\textwidth}
      \centering
      \input{fig/mpa.IndependentCascade.414-edges-txt.10}
      \caption{ICM}
    \end{subfigure}
    \begin{subfigure}{0.24\textwidth}
      \centering
      \input{fig/mpa.LinearThreshold.414-edges-txt.10}
      \caption{LTM}
    \end{subfigure}
    \begin{subfigure}{0.24\textwidth}
      \centering
      \input{fig/mpa.DeflatedIndependentCascade.414-edges-txt.10}
      \caption{DICM}
    \end{subfigure}
    \begin{subfigure}{0.24\textwidth}
      \centering
      \input{fig/mpa.GoyalContagion.414-edges-txt.10}
      \caption{SCM}
    \end{subfigure}

    \begin{subfigure}{0.24\textwidth}
      \centering
      \input{fig/mpa.IndependentCascade.414-edges-txt.15}
      \caption{ICM}
    \end{subfigure}
    \begin{subfigure}{0.24\textwidth}
      \centering
      \input{fig/mpa.LinearThreshold.414-edges-txt.15}
      \caption{LTM}
    \end{subfigure}
    \begin{subfigure}{0.24\textwidth}
      \centering
      \input{fig/mpa.DeflatedIndependentCascade.414-edges-txt.15}
      \caption{DICM}
    \end{subfigure}
    \begin{subfigure}{0.24\textwidth}
      \centering
      \input{fig/mpa.GoyalContagion.414-edges-txt.15}
      \caption{SCM}
    \end{subfigure}

    \begin{subfigure}{0.24\textwidth}
      \centering
      \input{fig/mpa.IndependentCascade.414-edges-txt.20}
      \caption{ICM}
    \end{subfigure}
    \begin{subfigure}{0.24\textwidth}
      \centering
      \input{fig/mpa.LinearThreshold.414-edges-txt.20}
      \caption{LTM}
    \end{subfigure}
    \begin{subfigure}{0.24\textwidth}
      \centering
      \input{fig/mpa.DeflatedIndependentCascade.414-edges-txt.20}
      \caption{DICM}
    \end{subfigure}
    \begin{subfigure}{0.24\textwidth}
      \centering
      \input{fig/mpa.GoyalContagion.414-edges-txt.20}
      \caption{SCM}
    \end{subfigure}

    \caption{\textbf{Evaluation of \algo{MPA} on an ego-network.} We evaluate the effectiveness of
    \algo{MPA}
    by comparing the results to \algo{DPIM}.  The network is another ego-network
    from Facebook that has 150 nodes and 3,386 edges.
    For each plot, the x-axis is $k$ and the y-axis is the expected total
    influence from the seed set chosen by each respective algorithm.}
    \label{fig:mpa_eval_rw}
  \end{figure*}

%
%
%
\section{Conclusions}
\label{sec:future}
We have given a heuristic which exploits the hierarchical community structure of
networks to find influential seed sets.  We have shown, using both real-world
and synthetic networks, that our algorithm outperforms the state of the art,
with large gains for non-submodular influence maximization.  We have also
exhibited ``worst-case" theoretical instances where our algorithm produces sets
that are $\Theta(\sqrt{n})$ more influential.  Lastly, we have generalized our
heuristic to a message passing algorithm.


One possible direction of future exploration is to try additional hierarchical decomposition techniques.  Interestingly, \algo{DPIM} can be seen as a way to \emph{test} hierarchical decomposition techniques.  Decompositions that perform better are intuitively finding a better decomposition.

\algo{DPIM} is typically not as fast as naive greedy, and to be useful in practice, it would greatly help if it were more scalable.  We believe that this will prove to be the case.  For example, we could stop the recursion before exploring the entire hierarchical decomposition.  We might stop dividing if a subtree does not appear to have any additional community structure, and then run a heuristic (such a degree or greedy) to process the rest of the subtree.  The intuition here is that dynamic programming works best where the network has strong community structure, so where no structure exists, it may not provide much added benefit.  Additionally, the same techniques that have made the greedy algorithm more scalable might be adopted to our dynamic programming and message passing frameworks.

%
%

%
\bibliography{grant}
\bibliographystyle{plain}

\end{document}